\documentclass[12pt,titlepage]{amsart}

\usepackage[T1]{fontenc}
\usepackage[utf8]{inputenc}
\usepackage[english]{babel}
\usepackage[foot]{amsaddr}
\usepackage{amssymb,amsmath,amsfonts,amsthm}
\usepackage{mathtools} 
\usepackage{enumerate,enumitem}
\usepackage{graphicx}
\usepackage{natbib,verbatim,subfigure}
\usepackage{tikz}
\usetikzlibrary{decorations,calc,intersections,shapes,spy}

\usepackage[left=1.25in,top=1.25in,right=1.25in,bottom=1.25in]{geometry}
\usepackage[onehalfspacing]{setspace}

\newcommand{\g}{\ifnum\currentgrouptype=16 \;\middle|\;\else\mid\fi}
\newcommand{\df}[1]{\textit{#1}}
\newcommand{\norm}[1]{\| #1 \|}

\newcommand{\abs}[1]{ \left | #1 \right | }

\def\Rbmon{\textrm{R}^{\textrm{mon}}}

\newcommand*\diff{\mathop{}\!\mathrm{d}}

\def\Re{\mathbf{R}} 
\def\Na{\mathbf{N}} 
 
\def\Qe{\mathbf{Q}}

\def\E{\mathop \mathbf{E}}

\def\ep{\varepsilon}
 
\def\ta{\theta}
\def\al{\alpha}
\def\la{\lambda}
\def\da{\delta}
\def\g{\gamma} 
\def\phi{\varphi}
\def\sa{\sigma}

\def\one{\mathbf{1}}

\def\os{\emptyset}

\def\Li{\textrm{Li}}
\def\Ls{\textrm{Ls}}
\def\Rbmon{\textrm{R}^{\textrm{mon}}}

\def\Pr{\mathop{}\mathrm{Pr}}

\def\ep{\varepsilon}

\def\os{\varnothing}
\def\one{\mathbf{1}}

\def\U{\mathcal{U}}
\def\P{\mathcal{P}}

\def\Vol{\mathrm{Vol}}

\newdimen\slantmathcorr
\def\oversl#1{%assuming that mathslant=0.25
\setbox0=\hbox{$#1$}
\slantmathcorr=\wd0
\hskip 0.2\slantmathcorr \overline{\hbox to 0.8\wd0{%
\vphantom{\hbox{$#1$}}}}
\hskip-\wd0\hbox{$#1$}
}
\def\undersl#1{%assuming that mathslant=0.25
\setbox0=\hbox{$#1$}
\slantmathcorr=\wd0
\underline{\hbox to 0.8\wd0{%
\vphantom{\hbox{$#1$}}}}
\hskip-0.8\wd0\hbox{$#1$}
}

\theoremstyle{plain}
\newtheorem{theorem}{Theorem}

\newtheorem{proposition}{Proposition}
\newtheorem{lemma}{Lemma}

\newtheorem{corollary}{Corollary}

\theoremstyle{definition}
\newtheorem{definition}{Definition}
\newtheorem{example}{Example}

\theoremstyle{remark}

\newtheorem*{claim*}{Claim}

\begin{document}
\title[Recovering utility]{
Recovering utility
%Preference Identification (or: Recovering utility functions? or: How I learned to stop worrying and love utility functions)
}% \\ \textit{\small\insertprelim{}}}

\author[Chambers]{Christopher P. Chambers}
\address[Chambers]{Department of Economics, Georgetown University}
\author[Echenique]{Federico Echenique}
\address[Echenique]{Department of Economics, UC Berkeley}
\author[Lambert]{Nicolas S. Lambert}
\address[Lambert]{Department of Economics, University of Southern California}

%\author[Chambers, Echenique and Lambert]{Christopher P. Chambers \and Federico Echenique \and Nicolas S. Lambert}

\thanks{Echenique thanks the National Science Foundation for its support through grant SES 1558757.
Lambert gratefully acknowledges the financial support and hospitality of Microsoft Research New York and the Yale University Cowles Foundation.}
%\date{First version: April 10, 2017; This version: August 18, 2018.}
\begin{abstract}
We provide sufficient conditions under which a utility function may be recovered from a finite choice experiment. Identification, as is commonly understood in decision theory, is not enough. We provide a general recoverability result that is widely applicable to modern theories of choice under uncertainty.  Key is to allow for a monetary environment, in which an objective notion of monotonicity is meaningful. In such environments, we show that subjective expected utility, as well as variational preferences, and other parametrizations of utilities over uncertain acts are recoverable. We also consider utility recovery in a statistical model with noise and random deviations from utility maximization.
\end{abstract}
\maketitle

\section{Introduction}

Economists are often interested in recovering preferences and utility functions from data on agents' choices. If we are able to recover a utility function, then a preference relation is obviously implied, but the inverse procedure is more delicate. In this paper, we presume access to data on an agent's choices, and that these describe the agent's preferences (or that preferences have been obtained as the outcome of a statistical estimation procedure). Our results describe sufficient conditions under which one can recover, or learn, a utility function from the agents' choices.

At a high level, the problem is that preferences essentially \textit{are} choices, because they encode the choice that would be made from each binary choice problem. When we write $x\succ y$ we really mean that $x$ would be chosen from the set $\{x,y\}$. Utility functions are much richer objects, and a given choice behavior may be described by many different utilities. For example, one utility can be used to discuss an agent's risk preferences: they could have a ``constant relative risk aversion'' utility, for which a single parameter describes attitudes towards risk. But the same preferences can be represented by a utility that does not have such a convenient parametrization. So recovering, or learning, utilities present important challenges that go beyond the problem of recovering a preference. In the paper, we describe some simple examples that illustrate the challenges. Our main results describe when one may (non-parametrically) recover a utility representation from choice data.

We first consider choice under uncertainty. We adopt the standard (Anscombe-Aumann) setting of choice under uncertainty, and focus attention on a class of utility representations that has been extensively studied in the literature. Special cases include subjected expected utility, the max-min expected utility model of \cite{gilboa1989}, Choquet expected utility \citep{schmeidler1989},  the variational preferences of \cite{maccheroni2006ambiguity}, and many other popular models. Decision theorists usually place significance on the uniqueness of their utility representations, arguing that uniqueness provides an identification argument that allows for utility to be recovered from choice data. We argue, in contrast, that uniqueness of a utility representation is \textit{not enough} to recover a utility from finite choice data.

Counterexamples are not hard to find. Indeed, even when a utility representation is unique, one may find a convergent sequence of utilities that is consistent with larger and larger finite datasets, but that does not converge to the utility function that generated the choices in the data, or to any utility to which it is equivalent. So uniqueness is necessary but not sufficient for a utility representation to be empirically tractable, in the sense of ensuring that a utility is recovered from large, but finite, choice experiments. 

Our main results are positive, and exhibit sufficient conditions for utility recovery. Key to our results is the availability of an objective direction of improvements in utility: we focus our attention on models of monotone preferences. Our paper considers choices among monetary acts, meaning state-contingent monetary payoffs. For such acts, there is a natural notion of monotonicity. Between two acts, if one pays more in every state of the world, the agent agent should prefer it. As a discipline on the recovery exercise, this essential notion of monotonicity suffices to ensure that a sequence of utilities that explains the choices in the data converges to the utility function that generated the choices.

We proceed by first discussing the continuity of a utility function in its dependence on the underlying preference relation. If $U(\succeq,x)$ is a function of a preference $\succeq$ and of choice objects $x$, then we say that it is a utility function if $x\mapsto U(\succeq,x)$ represents $\succeq$. We draw on the existing literature (Theorem~\ref{thm:levin}) to argue that such continuous utilities exist in very general circumstances. Continuity of this mapping in the preference ensures that if the choice data allow for preference recovery, they also allow a utility to be recovered. The drawback, however, of such general utility representation results is that they do not cover the special theories of utility in which economists generally take interest.  There is no reason to expect that the utility $U(\succeq,x)$ coincides with the standard parametrizations of, for example, subjective expected utility or variational preferences. 

We then go on to our main exercise, which constrains the environment to the Anscombe-Aumann setting, and considers utility representations that have received special attention in the theory of choice under uncertainty. We consider a setup that is flexible enough to accommodate most theories of choice under uncertainty that have been studied in the literature. Our main result (Theorem~\ref{thm:convergenceutilities}) says that, whenever a choice experiment succeeds in recovering agents' underlying preferences, it also serves to recover a utility in the class of utilities of interest. For example, if an agent has subjective expected utility preferences, and these can be recovered from a choice experiment, then so can the parameters of the subjective expected utility representation: the agents' beliefs and Bernoulli utility index. Or, if the agent has variational preferences that can be inferred from choice data, then so can the different components of the variational utility representation.

Actual data on choices may be subject to sampling noise, and agents who randomly deviate from their preferences. The results we have just mentioned are useful in such settings, once the randomness in preference estimates is taken into account. As a complement to our main findings, we proceed with a model that explicitly takes noisy choice, and randomness, into account. Specifically, we consider choice problems that are sampled at random, and an agent who may deviate from their preferences. They make mistakes. In such a setting, we present sufficient conditions for the consistency of utility function estimates (Theorem~\ref{thm:noisychoice}).  

In the last part of the paper we take a step back and revisit the problem of preference recovery, with the goal of showing how data from a finite choice experiment can approximate a preference relation, and, in consequence, a utility function. Our model considers a large, but finite, number of binary choices. We show that when preferences are monotone, then preference recovery is possible (Theorem~\ref{thm:weaklymon-v1}). In such environments, utility recovery follows for the models of choice under uncertainty that we have been interested in (Corollary~\ref{cor:connection}). 

\paragraph{Related literature.}

The literature on revealed preference theory in economics is primarily devoted to tests for consistency with rational choice. The main result in the literature, Afriat's theorem \citep{afriat,diewert1973afriat,varian1982nonparametric}, is in the context of standard demand theory (assuming linear budgets and a finite dataset). Versions of Afriat's result have been obtained in a model with infinite data \citep{reny2015characterization}, nonlinear budget sets (e.g., \citealp{matzkin1991axioms,forges2009}), general choice problems (e.g., \citealp{chavas1993generalized,nishimura}), and multiperson equilibrium models (e.g., \citealp{brown1996testable,carvajal2013revealed}).  Algorithmic questions related to revealed preference are discussed by \cite{echenique2011revealed} and \cite{camarachoice}. The monograph by \cite{chambers2016revealed} presents an overview of results. 

The revealed preference literature is primarily concerned with describing the datasets that are consistent with the theory, not with recovering or learning a preference, or a utility. In the context of demand theory and choice from linear budgets, \citet{mascolell78} introduces sufficient conditions under which a preference relation is recovered, in the limit, from a sequence of ever richer demand data observations. More recently, \citet{forges2009} derive the analog of \citeauthor{mascolell78}'s results for nonlinear budget sets. An important strand of literature focuses on non-parametric econometric estimation methods applied to demand theory data: \citet{blundell2003nonparametric,blundell2008best} propose statistical tests for revealed preference data, and consider counterfactual bounds on demand changes.

The problem of preference and utility recovery has been studied from the perspective of statistical learning theory. \cite{beigman2006learning} considers the problem of learning a demand function within the PAC paradigm, which is closely related to the exercise we perform in Section~\ref{sec:noise}. A key difference is that we work with data on pairwise choices, which are common in experimental settings (including in many recent large-scale online experiments).  \cite{zadimoghaddam2012efficiently} look at the utility recovery problem, as in \cite{beigman2006learning}, but instead of learning a demand function they want to understand when a utility can be learned efficiently. \cite{balcan2014learning} follow up on this important work by providing sample complexity guarantees, while \cite{ugarte2022preference} considers the problem of recovery of preferences under noisy choice data, as in our paper, but within the demand theory framework. Similarly, the early work of \cite{balcan2012learning} considers a PAC learning question, focusing on important sub-classes of valuations in economics. \cite{Bei2016LearningMP} pursues the problem assuming that a seller proposes budgets with the objective of learning an agent's utility (they focus on quasilinear utility, and a seller that obtains aggregate demand data).  \cite{Zhang2020LearningTV} considers this problem under an active-learning paradigm, and contrasts with the PAC sample complexity.

In all, these works are important precedents for our paper, but they are all within the demand theory setting. The results do not port to other environments, such as, for example, binary choice under risk or uncertainty. The closest paper to ours is \cite{chambers2021recovering}, which looks at a host of related questions to our paper but focusing on {\em preference}, not {\em utility}, recovery. The work by \citeauthor{chambers2021recovering} considers choices from binary choice problem, but does not address the question of recovering, or learning, a utility function. As we explain below in the paper, the problem for utilities is more delicate than the problem for preferences.  In this line of work, \cite{prasad2019} obtains important results on learning a utility but restricted to settings of intertemporal choice. The work by \cite{basu2020falsifiability} looks at learnability of utility functions (within the PAC learning paradigm), but focusing on particular models of choice under uncertainty. Some of our results rely on measures of the richness of a theory, or of a family of preferences, which is discussed by \cite{basu2020falsifiability} and \cite{fudenberg2021flexible}: the former by estimating the VC dimension of theories of choice under uncertainty, and the latter by proposing and analyzing new measures of richness that are well-suited for economics, as well as implementing them one economic datasets.

Finally, it is worth mentioning that preference and utilty recovery is potentially subject to to strategic manipulations, as emphasized by \cite{dong2018strategic} and \cite{echenique2020incentive}. This possibility is ignored in our work.

\section{The Question}\label{sec:thequestion}
We want to understand when utilities can be recovered from data on an agent's choices. Consider an agent with a utility function $u$. We want know when, given enough data on the agent's choices, we can ``estimate'' or ``recover'' a utility function that is guaranteed to be close to $u$. 

In statistical terminology, recovery is analogous to the consistency of an estimator, and approximation guarantees are analogous to learnability.  Imagine a dataset of size $k$, obtained from an incentivized experiment with $k$ different choice problems.\footnote{Such datasets are common in experimental economics, including cases with very large $k$. See, for example, \citet{vonGaudecker2011}, \citet{chapman2017willingness}, \citet{chapman2018econographics} and \citet{falkQJE2018}. One can also apply our results to roll call data from congress, as in  \cite{poolerosenthal1985} or \cite{clinton_jackman_rivers_2004}. Large-scale A/B testing by tech firms may provide further examples (albeit involving proprietary datasets).} The observed choice behavior in the data may be described  by a preference $\succeq^k$, which is associated with a utility function $u^k$. The preference $\succeq^k$ could be a rationalizing preference, or a preference estimate. So we choose a utility representation for $u^k$. The recovery, or consistency, property is that $u^k\to u$ as $k\to \infty$.

Suppose that the utility $u$ represents preferences $\succeq$, which summarize the agent's full choice behavior. Clearly, unless $\succeq^k\to\succeq$, the exercise is hopeless. So our first order of business is to understand when $\succeq^k\to\succeq$ is enough to ensure that $u^k\to u$. In other words, we want to understand when recovering preferences is sufficient for recovering utilities. To this end, our main results are in Section~\ref{sec:AAresults}. In recovering a utility, we are interested in particular parametric representations. In choice over uncertainty, for example, one may be interested in measures of risk-attitudes, or uncertainty aversion. It is key then that the utility recovery exercises preserves the aspects of utility that allow such measures to be have meaning. If, say, preferences have the ``constant relative risk aversion'' (CRRA) form, then we want to recover the Arrow-Pratt measure of risk aversion. 

Our data is presumably obtained in an experimental setting, where an agent's behavior may be recorded with errors; o in which the agent may randomly deviate from their underlying preference $\succeq$. Despite such errors,  with high probability, ``on the sample path,''  we should obtain that  $\succeq^k\to\succeq$. In our paper we uncover situations where this convergence leads to utility recovery. Indeed, the results in Section~\ref{sec:AAresults} and~\ref{sec:certaintyequiv} may be applied to say that, in many popular models in decision theory, when $\succeq^k\to\succeq$ (with high probability), then the resulting utility representations enable utility recovery (with high probability).

The next step is to discuss learning and sample complexity. Here we need to explicitly account for randomness and errors. We lay out a model of random choice, with random sampling of choice problems and errors in agents' choices. The errors may take a very general form, as long as random choices are more likely to go in the direction of preferences than against it (if $x\succ y$ then $x$ is the more likely choice from the choice problem $\{x,y\}$), and that this likelihood ratio remains bounded away from one. Contrast with the standard theory of discrete choice, where the randomness usually is taken to be additive, and independent of the particular pair of alternatives that are being compared. 

Here we consider a formal statistical consistency problem, and exhibit situations where utility recovery is feasible. We use ideas from the literature on PAC learning to provide formal finite sample-size bounds for each desired approximation guarantee. See Section~\ref{sec:noise}.

\section{The Model}\label{sec:model}

\subsection{Basic definitions and notational conventions}
Let $X$ be a set.  Given a binary relation $R\subseteq X\times X$, we write $x \mathrel{R} y$ when $(x,y)\in R$.  A binary relation that is complete and transitive is called a \df{weak order}.  If $X$ is a topological space, then we say that $R$ is \df{continuous} if $R$ is closed as a subset of $X\times X$ (see, for example, \citealp{bergstrom1976}). A \df{preference relation} is a weak order that is also continuous. 

A preference relation $\succeq$ is \df{locally strict} if, for all $x,y\in X$, $x \succeq y$ implies that for each neighborhood $U$ of $(x,y)$, there is $(x',y')\in U$ with $x \succ y$.  The notion of local strictness was first introduced by \citet{border1994dynamic} as a generalization of the property of being locally non-satiated from consumer theory.

If $\succeq$ is a preference on $X$ and $u:X\to \Re$ is a function for which $x\succeq y$ if and only if $u(x)\geq u(y)$ then we say that $u$ is a \df{representation} of $\succeq$, or that $u$ is a \df{utility function} for $\succeq$.

If $A\subseteq \Re^d$ is a Borel set, we write $\Delta(A)$ for the set of all Borel probability measures on $A$. We endow $\Delta(A)$ with the weak* topology. If $S$ is a finite set, then we topologize $\Delta(A)^S$ with the product topology. 

For $p,q\in\Delta(A)$, we say that $p$ is larger than $q$ in the sense of \df{first-order stochastic dominance} if $\int_A fdx\geq \int_A fdy$ for all monotone increasing, continuous and bounded functions $f$ on $A$.

\subsection{Topologies on preferences and utilities.}\label{sec:closedconvergence}

The set of preferences over $X$, when $X$ is a topological space, is endowed with the topology of closed convergence. The space of corresponding utility representations is endowed with the compact-open topology. These are the standard topologies  for preferences and utilities, used in prior work in mathematical economics. See, for example, \cite{HILDENBRAND1970161}, \cite{kannai1970continuity}, and \cite{mas1974continuous}. Here we offer definitions and a brief discussion of our choice of topology.

Let $X$ be a topological space, and $\mathcal{F}=\{F^n\}_n$ be a sequence of closed sets in $X \times X$ (with the product topology). We define $\Li(\mathcal{F})$ and $\Ls(\mathcal{F})$ to be closed subsets of $X \times X$ as follows:
\begin{itemize}
	\item $(x,y) \in \Li(\mathcal{F})$ if and only if, for all neighborhoods $V$ of $(x,y)$, there exists $N \in \Na$ such that $F^n \cap V \neq \os$ for all $n \geq N$.
	\item $(x,y) \in \Ls(\mathcal{F})$ if and only if, for all neighborhoods $V$ of $(x,y)$, and all $N \in \Na$, there is  $n \geq N$ such that $F^n \cap V \neq \os$.
\end{itemize}
Observe that $\Li(\mathcal{F})\subseteq \Ls(\mathcal{F})$. The definition of closed convergence is as follows.
\begin{definition}
	$F^n$ converges to $F$ in the \df{topology of closed convergence} if $\Li(\mathcal{F})=F=\Ls(\mathcal{F})$.
\end{definition}

Closed convergence  captures the property that agents with similar preferences should have similar choice behavior---a property that is necessary to be able to learn the preference from finite data. Specifically, if $X\subseteq \Re^n$, and $\P$ is the set of all locally strict and continuous preferences on $X$, then the topology of closed convergence is the smallest topology on $\P$ for which the sets
\[
	\{ (x,y,\succeq) : x\succ y\}\subseteq X\times X\times \P
\]
are open.\footnote{See \cite{kannai1970continuity} and \cite{HILDENBRAND1970161} for a discussion; a proof of this claim is available from the authors upon request.} In words: suppose that $x\succ y$, then for $x'$ close to $x$, $y'$ close to $y$, and $\succeq'$ close to $\succeq$, we obtain that $x'\succ' y'$. 

For utility functions, we adopt the compact-open topology, which we also claim is a natural choice of topology. The compact-open topology is characterized by the convergence criterion of uniform convergence on compact sets. The reason it is natural for utility functions is that a utility usually has two arguments: one is the object being ``consumed'' (a lottery, for example) and the other is the ordinal preference that utility is meant to represent. (The preference argument is usually implicit, but of course it remains a key aspect of the exercise.) Now an analyst wants the utility to be ``jointly continuous,'' or continuous in both of its arguments. For such a purpose, the natural topology on the set of utilities, when they are viewed solely as functions of consumption, is indeed the compact-open topology. More formally, consider the  following result, originally due to \cite{mascolell1977IER}.\footnote{\cite{levin1983continuous} provides a generalization to incomplete preferences.}

\begin{theorem}\label{thm:levin}
    Let $X$ be a locally compact Polish space, and $\P$ the space of all continuous preferences on $X$ endowed with the topology of closed convergence. Then there exists a continuous function $U:\P\times X\to [0,1]$ so that $x\mapsto U(\succeq,x)$ represents $\succeq$.
\end{theorem}

We may view the map $U$ as a mapping from $\succeq$ to the space of utility functions.  Then continuity of this induced mapping is equivalent to the joint continuity result discussed in Theorem~\ref{thm:levin}, as long as we impose the compact-open topology on the space of utility functions (see \cite{foxBAMS1945}).

\subsection{The model}\label{sec:themodel}
As laid our in Section~\ref{sec:thequestion}, we want to understand when we may conclude that $u^k\to u$ from knowing that $\succeq^k\to\succeq$.  Mas-Colell's theorem (Theorem~\ref{thm:levin}) provides general conditions under which there exists {\em one} utility representation that has the requisite convergence property, but he is clear about the practical limitations of his result:  ``There is probably not a simple constructive (``canonical'') method to find a $U$ function.''  In contrast, economists are generally interested in {\em specific} parameterizations of utility.  

For example, if an agent has subjective expected-utility preferences, economists want to estimate beliefs and a von-Neumann-Morgenstern index; not some arbitrary representation of the agent's preferences. Or, if the data involve intertemporal choices, and the agent discounts utility exponentially, then an economist will want to estimate their discount factor. Such specific parameterizations of utility are not meaningful in the context of Theorem~\ref{thm:levin}.

The following (trivial) example shows that there is indeed a problem to be studied. Convergence of arbitrary utility representations to the correct limit is not guaranteed, even when recovered utilities form a convergent sequence, and recovered preferences converge to the correct limit.

\begin{example} Consider expected-utility preferences on $\Delta(K)^S$, where $K$ is a compact space, $S$ a finite set of states, and $\Delta^S(K)$ is the set of Anscombe-Aumann acts. Fix an affine function $v:\Delta(K)\to\Re$, a prior $\mu\in\Delta(S)$, and consider the preference $\succeq$ with representation $\int_S v(f(s))\diff \mu(s)$. 

Now if we set $\succeq^k=\succeq$ then $\succeq^k\to\succeq $ holds trivially. However, it is possible to choose an expected utility representation $\int_S v^k(f(s))\diff \mu^k(s)$ that does not converge to a utility representation (of any kind) for $\succeq$. In fact one could choose a $\mu^k$ and a ``normalization'' for $v^k$, for example $\norm{v^k}=1$ (imagine for concreteness that $K$ is finite, and use the Euclidean norm for $v^k$). Specifically, choose scalars $\beta^k$ with $\norm{\beta^k + \frac{1}{k}v}=1$. Then the utility $f\mapsto \int_S v^k(f(s))\diff \mu(s)$ represents $\succeq^k$ and converges to a constant function. 

The punchline is that the limiting utility represents the preference that exhibits complete indifference among all acts. This is true, no matter what the original preference $\succeq$ was.
\end{example}

In the example, we have imposed some discipline on the representation. Given that the utility converges to a constant, the discipline we have chosen is a particular normalization of the utility representations  (their norm is constant). The normalization just makes the construction of the example slightly more challenging, and reflects perhaps the most basic care that an analyst could impose on the recovery exercise.

\subsection{Anscombe-Aumann acts}\label{sec:AAresults}

We present our first main result in the context of Anscombe-Aumann acts, the workhorse model of the modern theory of decisions under uncertainty. Let $S$ be a finite set of \df{states of the world}, and fix a closed interval of the real line $[a,b]\subseteq \Re$. An \df{act} is a function $f:S\to \Delta([a,b])$. We interpret the elements of $\Delta([a,b])$ as \df{monetary lotteries}, so that acts are state-contingent monetary lotteries. The set of all acts is  $\Delta([a,b])^S$. When $p\in\Delta([a,b])$, we denote the \df{constant act} that is identically equal to $p$ by $(p,\ldots,p)$; or sometimes by $p$ for short.

Note that we do not work with abstract, general, Anscombe-Aumann acts, but in assuming monetary lotteries we impose a particular structure on the objective lotteries in our Anscombe-Aumann framework. The reason is that our theory necessitates a certain {\em known and objective} direction of preference. Certain preference comparisons must be known {\em a priori}: monotonicity of preference will do the job, but for monotonicity to be objective we need the structure of monetary lotteries.

An act $f$ \df{dominates} an act $g$ if, for all $s\in S$, $f(s)$ first-order stochastic dominates $g(s)$.  And $f$ \df{strictly dominates}  $g$ if, for all $s\in S$, $f(s)$ strictly first-order stochastic dominates $g(s)$. A preference $\succeq$ over acts is \df{weakly monotone} if $f\succeq g$ whenever $f$ first-order stochastic dominates $g$.

Let $U$ be the set of all continuous and monotone weakly increasing functions $u:[a,b]\to\Re$ with $u(a)=0$ and $u(b)=1$.  A pair $(V,u)$ is a \df{standard representation} if $V:\Delta([a,b])^S\to\Re$ and $u\in U$ are continuous functions such that $v(p,\ldots,p)=\int_{[a,b]}u \diff p$, for all constant acts $(p,\ldots,p)$. Moreover, we say that a standard representation $(V,u)$ is \df{aggregative} if there is an \df{aggregator} $H:[0,1]^S\to\Re$ with $V(f)=H((\int u \diff f(s))_{s\in S})$ for $f\in \Delta([a,b])^S$. An aggregative representation with aggregator $H$ is denoted by $(V,u,H)$.  Observe that a standard representation rules out total indifference.

A preference $\succeq$ on $\Delta([a,b])^S$ is \df{standard} if it is weakly monotone, and there is a standard representation $(V,u)$ in which $V$ represents $\succeq$.  Roughly, standard preferences will be those that satisfy the expected utility axioms across constant acts, and are monotone with respect to the (statewise) first order stochastic dominance relation.  Aggregative preferences will additionally satisfy an analogue of Savage's P3 or the Anscombe-Aumann notion of monotonicity. 

\begin{example}
Variational preferences \citep{maccheroni2006ambiguity} are standard and aggregative.\footnote{Variational preferences are widely used in macroeconomics and finance to capture decision makers' concerns for using a misspecified model. Here it is important to recover the different components of a representation, $v$ and $c$, because they quantify key features of the environment. See for example \cite{hansen2001robust,HANSEN200645,hansen2022risk}.}
Let \[ 
V(f) = \inf\{ \int v(f(s)) d\pi (s) + c(\pi) : \pi\in\Delta(S) \}
\]
where 
\begin{enumerate}
\item $v:\Delta([a,b])\to\Re$ is continuous and affine.
    \item $c:\Delta(S)\to [0,\infty]$ is lower semicontinuous, convex and grounded (meaning that $\inf \{c(\pi):\pi\in \Delta(S) \}=0$). 
\end{enumerate} 

Note that $V(p,\ldots,p)= v(p) + \inf\{ c(\pi) : \pi\in\Delta(S) \} = \int u\diff p$, by the assumption that $c$ is grounded, and where the existence of $u:[a,b]\to\Re$ so that $v(p)=\int u\diff p$ is an instance of the Riesz representation theorem. It is clear that we may choose $u\in U$. So $(V,u)$ is a standard representation.

Letting $H:[0,1]^S\to \Re$ be defined by $H(x) = \inf\{ \sum_{s\in S} x(s) \pi(s) + c(\pi) : \pi\in\Delta(S) \}$, we see that indeed $(V,u,H)$ is also an aggregative representation of these preferences.
\end{example}

Some other examples of aggregative preferences include special cases of the variational model \citet{gilboa1989}, as well as generalizations of it, \citet{cerreia2011,chandrasekher2021}, and others which are not comparable \citet{schmeidler1989,chateauneuf2008,chateauneuf2009}.\footnote{A class of variational preferences that are of particular interest to computer scientists are preferences with a max-min representation \citep{gilboa1989}. These evaluate acts by 
\[ 
V(f) = \inf\{ \int v(f(s)) d\pi (s) : \pi\in\Pi \},
\] with $\Pi\subseteq \Delta(S)$ a closed and convex set. Here $c$ is the indicator function of $\Pi$ (as defined in convex analysis).
}

\begin{theorem}\label{thm:convergenceutilities} Let $\succeq$ be a standard preference with standard representation $(V,u)$, and $\{\succeq^k\}$ a sequence of standard preferences, each with a standard representation $(V^k,u^k)$. \begin{enumerate}
    \item If $\succeq^k\to \succeq$, then $(V^k,u^k)\to (V,u)$.
    \item If, in addition, these preferences are aggregative with representations $(V^k,u^k,H^k)$ and $(V,u,H)$, then $H^k\to H$.
\end{enumerate}
\end{theorem}

In terms of interpretation, Theorem~\ref{thm:convergenceutilities} suggests that, as preferences converge, risk-attitudes, or von Neumann morgenstern utility indices also converge in a pointwise sense.  The aggregative part claims that we can study the convergence of risk attitudes and the convergence of the aggregator controlling for risk separately.  So, for example, in the multiple priors case, two decision makers whose preferences are close will have similar sets of priors.

\subsection{Preferences over lotteries and certainty equivalents}\label{sec:certaintyequiv}

In this section, we focus on a canonical representation for preferences over lotteries:  the certainty equivalent.  There are many models of preferences over lotteries, but we have in mind in particular \cite{cerreia2015cautious}, whereby a preference representation over lotteries is given by $U(p) = \inf_{u\in \mathcal{U}}u^{-1}(\int u dp)$; a minimum over a set of certainty equivalents for expected utility maximizers.  Key is that for this representation, and any degenerate lottery $\delta_x$, $U(\delta_x)=x$.  

Let $[a,b]\subset \Re$, where $a<b$, be an interval in the real line and consider $\Delta([a,b])$.  Say that $\succeq$ on $\Delta([a,b])$ is \emph{certainty monotone} if when ever $p$ first order stochastically dominates $q$, then  $p \succeq q$, and for all $x,y\in [a,b]$ for which $x> y$, $\delta_x \succ \delta_y$.  Any certainty monotone continuous preference $\succeq$ and any lottery $p\in\Delta([a,b])$ then possesses a unique \df{certainty equivalent} $x\in [0,1]$, satisfying $\delta_x \sim p$.  To this end, we define $\mbox{ce}(\succeq,p)$ to be the certainty equivalent of $p$ for $\succeq$.  It is clear that, fixing $\succeq$, $\mbox{ce}(\cdot,\succeq)$ is a continuous utility representation of $\succeq$.  

\begin{proposition}\label{prop:certaintyequiv}Let $\succeq$ be a certainty monotone preference and let $p\in \Delta([a,b])$.  Let $\{\succeq^k\}$ be a sequence of certainty monotone preferences and let $p^k$ be a sequence in $\Delta([a,b])$.  If $(\succeq^k,p^k)\rightarrow (\succeq,p)$, then $\mbox{ce}(\succeq^k,p^k)\rightarrow \mbox{ce}(\succeq,p)$.\end{proposition}

To this end, the map carrying each preference to its certainty equivalent representation is a continuous map in the topology of closed convergence.

\section{Utility recovery with noisy choice data}\label{sec:noise}

We develop a model of noisy choice data, and consider when utility may be recovered from a traditional estimation procedure. Recovery here takes the form of an explicit consistency result, together with sample complexity bounds in a PAC learning framework.

The focus is on the \df{Wald representation}, analogous to the certainty equivalent we considered in Section~\ref{sec:certaintyequiv}. When choosing among vectors in $x\in \Re^d$, the Wald representation is $u(x)\in \Re$ so that 
\[ 
x\sim (u(x),\ldots,u(x)).
\]  If the choice space is well behaved, a Wald representation exists for any monotone and continuous preference relation. To this end, we move beyond the Anscombe-Aumann setting that we considered above, but it should be clear that some versions of Anscombe-Aumann can be accommodated within the assumptions of this section.

Our main results for the model that explicitly accounts for noisy choice data assumes Wald representations that are either Lipschitz or homogeneous (meaning that preferences are homothetic). 

\subsection{Noisy choice data}\label{sec:noisydata}

The primitives of our noisy choice model are collected in the tuple  $(X,\P,\la,q)$, where:
\begin{itemize}
    \item $X\subseteq\Re^d$ is the ambient choice, or consumption, space. The set $X$ is endowed with the (relative) topology inherited from $\Re^d$.
    \item $\P$ is a class of continuous and locally strict preferences on $X$. The class comes with a set of utility functions $\U$, so that each element of $\P$ has a utility representation in the set $\U$.
    \item $\la$ is a probability measure on $X$, assumed to be absolutely continuous with respect to Lebesgue measure. We also assume that $\la\geq c \,\mathrm{Leb}$, where $c>0$ is a constant and Leb denotes Lebesgue measure. 
    \item $q:X\times X\times \P\to [0,1]$ is a random choice function, so $q(x,y; \succeq)$ is the probability that an agent with preferences $\succeq$ chooses $x$ over $y$. Assume that if $x\succ y$, then $x$ is chosen with probability $q(x,y;\succeq)>1/2$ and $y$ with probability $q(y,x;\succeq^*)=1-q(x,y;\succeq)$. If $x\sim y$ then $x$ and $y$ are chosen with equal probability.
    \item We shall assume that the error probability $q$ satisfies that \[
\Theta\equiv \inf\{q(\succeq,(x,y)): x\succ y \text{ and } \succeq\in\P\} > \frac{1}{2}.
\]
\end{itemize}

The tuple $(X,\P,\la,q)$ describes a data-generating process for noisy choice data. Fix a sample size $n$ and consider an agent with preference $\succeq^*\in \P$. A sequence of choice problems $\{x_i,y_i\}$, $1\leq i\leq n$ are obtained by drawing $x_i$ and $y_i$ from $X$, independently, according to the law $\lambda$. Then a choice is made from each problem $\{x_i,y_i\}$ according to $q(\cdot,\cdot;\succeq^*)$.

Observe that our assumptions on $q$ are mild. We allow errors to depend on the pair $\{x,y\}$ under consideration, almost arbitrarily. The only requirement is that one is more likely to choose according to one's preference than to go against them, as well as the more technical assumptions of measurability and a control on how large the deviation from $1/2$-$1/2$ choice may get.

To keep track of the chosen alternative, we order the elements of each problem so that $(x_i,y_i)$ means that $x_i$ was chosen from the choice problem $\{x_i,y_i\}$. So a sample of size $n$ is $\{(x_1,y_1),\ldots,(x_n,y_n) \}$, consisting of $2n$ iid draws from $X\times X$ according to our stochastic choice model: in the $i$th draw, the choice problem was $\{x_i,y_i\}$ and $x_i$ was chosen.

A utility function $u_n\in \U$ is chosen to maximize the number of rationalized choices in the data. So $u_n$ maximizes $\sum_{i=1}^n \one_{u(x_i)\geq u(y_i)}$. The space of utility functions is endowed with a metric, $\rho$. In this section, all we ask of $\rho$ is that, for any $u,u'\in \U$, there is $x\in X$ with $\abs{u(x)-u'(x)}\geq \rho(u,u')$. For example, we could use the sup norm for the purposes of any of the results in this section.

\subsubsection{Lipschitz utilities}\label{sec:lipsch}
One set of sufficient conditions will need the family of relevant utility representations to satisfy a Lipschitz property with a common Lipschitz bound. The representations are of the Wald kind, as in Section~\ref{sec:certaintyequiv}. We now add the requirement of having the Lipschitz property, which allows us to connect differences in utility functions to quantifiable observable (but noisy) choice behavior. The main idea is expressed in Lemma~\ref{lem:Lipsch} of Section~\ref{sec:proofs}.

We say that $(X,\P,\la,q)$ is a \df{Lipschitz environment} if:
\begin{enumerate}
  \item $X\subseteq\Re^d$
is convex, compact, and has nonempty  interior. \item Each preference $\succeq\in \P$
has a Wald utility representation $u_\succeq:X\to\Re$ so that $x\sim u_\succeq
(x)\one$. \item All utilities in $\U$ are Lipschitz, and admit a common Lipschitz
constant $\kappa$. So, for any $x,x'\in X$ and $u\in \U$, $|u(x)-u(x')|\leq \kappa
\norm{x-x'}$.  \end{enumerate}

\subsubsection{Homothetic preferences}\label{sec:homoth}
The second set of sufficient conditions involve homothetic preferences. It turns out, in this case, that the Wald representations have a homogeneity property, and this allows us to connect differences in utilities to a probability of detecting such differences. The key insights is contained in Lemma~\ref{lem:homothetic} of Section~\ref{sec:proofs}.

We employ the following auxiliary notation. $S^M_\al = \{x\in \Re^d : \norm{x}=M \text{ and } x\geq \al\one\}$ and  $D^M_\al = \{\ta x:x\in S^M_\al \text{ and } \ta\in [0,1]\}$.

We say that $(X,\P,\la,q)$ is a \df{homothetic environment} if:
\begin{enumerate}
\item $X=D^M_\al$ for some (small) $\al>0$ and (large) $M>0$.
\item $\P$ is a class of continuous, monotone, homothetic, and complete preferences on
  $X\subseteq \Re^d$.
\item $\U$ is a class of Wald representations, so that for each $\succeq\in \P$
  there is a utility function $u\in \U$ with $x\sim u(x)\one$.
\end{enumerate}

Remark: if $u\in U$ is the Wald representation of $\succeq$, then $u$ is homogeneous of degree one because $x\sim u(x)\one$ iff $\la x \sim \la u(x)\one$, so $u(\la x)=\la u(x)$. 

\subsubsection{VC dimension} The Vapnik-Chervonenkis (VC) dimension of a set $\P$ of preferences is the largest sample size $n$ for which there exists a utility $u\in \U$ that perfectly rationalizes all the choices in the data, no matter what those are. That is so that $n=\sum_{i=1}^n \one_{u(x_i)\geq u(y_i)}$ for any dataset $(x_i,y_i)_{i=1}^n$ of size $n$.

VC dimension is a basic ingredient in the standard PAC learning paradigm. It is a measure of the complexity of a theory used in machine learning, and lies behind standard results on uniform laws of large numbers (see, for example, \cite{boucheron2005theory}). Applications of VC to decision theory can be found in \cite{basu2020falsifiability} and \cite{chambers2021recovering}. 

It is worth noting that VC dimension is used in classification tasks. It may not be obvious, but when it comes to preferences, our exercise may be thought of as classification. For each pair of alternatives $x$ and $y$, a preference $\succeq$ ``classifies'' the pair as $x\succeq y$ or $y\succ x$. Then we can think of preference recovery as a problem of learning a classifier within the class $\P$. 

\subsection{Consistency and sample complexity}

\begin{theorem}\label{thm:noisychoice}
    Consider a noisy choice environment $(X,\P,\la,q)$ that is either a homothetic or a Lipschitz environment. Suppose that $u^*\in \U$ is the Wald utility representation of~$\succeq^*\in \P$. 
  \begin{enumerate}
  \item The estimates $u_n$ converge to $u^*$ in probability.
    \item 
 There are constants $K$ and $\bar C$ so that, for any $\da\in (0,1)$ and $n$, with probability at least $1-\da$, 
 \[
\rho(u_n,u^*)\leq \bar C \left(K\sqrt{V/n} + \sqrt{2\ln (1/\da)/n} \right)^{1/D},
\] where $V$ is the VC dimension of $\P$, $D=d$ when the environment is Lipschitz and $D=2d$ when it is homothetic.
\end{enumerate}
\end{theorem}

Of course, the second statement in the theorem is only meaningful when the VC dimension of $\P$ is finite. The constants $K$ and $\bar C$ depend on the primitives in the environment, but not on preferences, utilities, or sample sizes.

\section{Recovering preferences and utilities}\label{sec:recoverpreferences}

The discussion in Section~\ref{sec:AAresults} focused on utility recovery, taking convergence of preferences as given. Here we take a step back, provide some conditions for preference recovery that are particularly relevant for the setting of Section~\ref{sec:AAresults}, and then connect these back to utility recovery in Corollary~\ref{cor:connection}. First we describe an experimental setting in which preferences may be elicited: an agent, or subject, faces a sequence of (incentivized) choice problems, and the choices made produce data on his preferences. The specific model and description below is borrowed from \cite{chambers2021recovering}, but the setting is completely standard in choice theory.

Let $X=\Delta([a,b])^S$ be the set of acts over monetary lotteries, as discussed in Section~\ref{sec:AAresults}.  A  \df{choice function} is a pair  $(\Sigma,c)$ with  $\Sigma\subseteq 2^X\setminus\{\os\}$  a collection of nonempty subsets of $X$, and $c:\Sigma\to 2^X$ with $\os\neq c(A)\subseteq A$ for all $A\in\Sigma$. When $\Sigma$, the domain of $c$, is implied, we refer to $c$ as a choice function.

A choice function $(\Sigma,c)$ is \df{generated} by a preference relation $\succeq$ over $X$ if \begin{equation*}
	c(A)= \{x\in A : x\succeq y \text{ for all } y\in B \},
\end{equation*}
for all $A \in \Sigma$.

The notation $(\Sigma,c_{\succeq})$ means that the choice function $(\Sigma,c_{\succeq})$ is generated by the preference relation $\succeq$ on $X$.

Our model features an experimenter (a female) and a subject (a male). The subject chooses among alternatives in a way described by a preference $\succeq^*$ over $X$, which we refer to as data-generating preference. The experimenter seeks to infer $\succeq^*$ from the subject's choices in a finite experiment.

In a finite experiment, the subject is presented with finitely many unordered pairs of alternatives $B_k=\{x_k,y_k\}$  in $X$. For every pair $B_k$, the subject is asked to choose one of the two alternatives: $x_k$ or $y_k$.

A \df{sequence of experiments} is a collection  $\Sigma_{\infty} = \{B_i\}_{i \in \Na}$ of pairs of possible choices presented to the subject. Let $\Sigma_k=\{B_1, \dots, B_k\}$ collect the first $k$ elements of a sequence of experiments, and  $B = \cup_{k=1}^\infty B_k$ be the set of all alternatives that are used over all the experiments in a sequence. Here $\Sigma_k$ is a finite experiment of \df{size} $k$.

We make two assumptions on $\Sigma_\infty$. The first is that $B$ is dense in $X$. The second is that, for any $x,y\in B$ there is $k$ for which $B_k=\{x,y\}$. The first assumption is obviously needed to obtain any general preference recovery result. The second assumption means that the experimenter is able to elicit the subject's choices over all pairs used in her experiment.\footnote{If there is a countable dense $A\subseteq X$, then one can always construct such a sequence of experiments via a standard diagonalization argument.}

For each $k$, the subject's preference $\succeq^*$ generates a choice function $(\Sigma_k,c)$ by letting, for each $B_i\in\Sigma_k$, $c(B)$ be a maximal element of $B_i$ according to $\succeq^*$. Thus the choice behavior observed by the experimenter is always consistent with $(\Sigma_k,c_{\succeq^*})$. 

We introduce two notions of rationalization: weak and strong.  A preference $\succeq_k$ \df{weakly rationalizes} $(\Sigma_k,c)$  if, for all $B_i\in \Sigma_k$,  $c(B_i)\subseteq c_{\succeq_k}(B_i)$. A preference $\succeq_k$ weakly rationalizes a choice sequence $(\Sigma_\infty,c)$ if it rationalizes the choice function of order $k$ $(\Sigma_k,c)$, for all $k \ge 1$.

A preference $\succeq_k$ \df{strongly rationalizes} $(\Sigma_k,c)$  if, for all $B_i\in \Sigma_k$,  $c(B_i)= c_{\succeq_k}(B_i)$. A preference $\succeq_k$ strongly rationalizes a choice sequence $(\Sigma_\infty,c)$ if it rationalizes the choice function of order $k$ $(\Sigma_k,c)$, for all $k \ge 1$.

In the history of revealed preference theory in consumer theory, strong rationalizability came first. It is essentially the notion in \citet{samuelson1938note} and \citet{richter}. Strong rationalizability is the appropriate notion when it is known that all potentially chosen alternatives are actually chosen, or when we want to impose, as an added discipline, that the observed choices are uniquely optimal in each choice problem. This makes sense when studying demand functions, as Samuelson did. Weak rationalizability was one of the innovations in \citet{afriat67}, who was interested in demand correspondences.\footnote{As an illustration of the difference between these two notions of rationalizability, note that, in the setting of consumer theory, one leads to the Strong Axiom of Revealed Preference while the other to the Generalized Axiom of Revealed Preference. Of course, Afriat's approach is also distinct in assuming a finite dataset. See \cite{chambers2016revealed} for a detailed discussion.}

\subsection{A general ``limiting'' result}

Our next result serves to contrast what can be achieved with the ``limiting'' (countably infinite) data with the limit of preferences recovered from finite choice experiments.

\begin{theorem}\label{thm:connected-v1}
 Suppose that $\succeq$ and $\succeq^*$ are two continuous preference relations (complete and transitive).
If $\succeq|_{B\times B} = \succeq^*|_{B\times B}$, then $\succeq = \succeq^*$.
\end{theorem}

Indeed, as the proof makes clear, Theorem~\ref{thm:connected-v1} would hold more generally for any $X$ which is a connected topological space, but it may not hold in absence of connectedness.   There is a sense in which the limiting case with an infinite amount of data offers no problems for preference recovery. The structure we impose is needed for the limit of rationalizations drawn from finite data.

\subsection{Recovery from finite data in the AA model}
Here we adopt the same structural assumptions as in Section~\ref{sec:AAresults}, namely that $X=\Delta([a,b])^S$, endowed with the weak topology and the first order stochastic dominance relation.  However, the result easily extends to  broader environments, as the proof makes clear.

\begin{theorem}\label{thm:weaklymon-v1} There is a sequence of finite experiments $\Sigma_\infty$ so that if the subject's preference $\succeq^*$ is continuous and weakly monotone, and for each $k\in\Na$,  $\succeq^k$ is a continuous and weakly monotone preference that strongly rationalizes a choice function $(\Sigma_k,c)$ generated by $\succeq^*$;  then $\succeq_k \rightarrow \succeq^*$.
\end{theorem}

\begin{corollary}\label{cor:connection}
Let $\succeq^*$ and $\succeq^k$ be as in the statement of Theorem~\ref{thm:weaklymon-v1}. If, in addition,  $\succeq^*$ and $\succeq^k$ have standard representations $(V,u)$ and $(V^k,u^k)$ then $(V,u)=\lim_{k\to\infty} (V^k,u^k)$.
\end{corollary}

Note that Theorem~\ref{thm:weaklymon-v1} requires the existence of the data-generating preference $\succeq^*$.   
 
A ``dual'' result to Theorem~\ref{thm:weaklymon-v1} was established in \citet{chambers2021recovering}.  There, the focus was on weak rationalization via $\succeq^k$, which is a weaker notion than the strong rationalization hypothesized here.  To achieve a weak rationalization result, we assumed instead that preferences were strictly monotone.

\begin{comment}
Fix a sequence $\Sigma_\infty$ of finite experiments, and suppose that $(\Sigma_k,c^k)$ is  a sequence of choices made for each finite experiment $k$.  Let
$\mathcal{P}^k(c^k)$ be the set of continuous and strictly monotone preferences that weakly rationalize $c^k$. For a set of binary relations $S$, define $\mbox{diam}(S)=\sup_{(\succeq,\succeq')\in S^2}\delta_C(\succeq,\succeq')$ to be the diameter of $S$ according to the metric $\delta_C$ which generates the topology on preferences.

\begin{theorem}\label{thm:diameter}
One of the following statements holds:
\begin{enumerate}
\item There is $k$ such that $\mathcal{P}^k(c)=\varnothing$.
\item $\lim_{k\rightarrow \infty}\mbox{diam}(\mathcal{P}^k(c))\rightarrow 0$.
\end{enumerate}
\end{theorem}

That is, either a choice sequence is eventually not weakly rationalizable by a strictly monotone preference, \emph{or}, the set of rationalizations becomes arbitrarily small.

Remark: Theorem~\ref{thm:diameter} can, in fact, dispense with the hypothesis of transitivity. In this case, we would define $\mathcal{P}^k(c^k)$ to be the set of (potentially nontransitive) complete, continuous, and strongly monotone relations weakly rationalizing $c^k$.
\end{comment}

\section{Proofs}\label{sec:proofs}

In this section, unless we say otherwise, we denote by $X$ the set of acts $\Delta([a,b])^S$, and the elements of $X$ by $x,y,z$ etc. Note that $X$ is compact Polish when $\Delta([a,b])$ is endowed with the topology of weak convergence of probability measures. Let $\P$ be the set of all complete and continuous binary relations on $X$.

\subsection{Lemmas}

The lemmas stated here will be used in the proofs of our results. 

\begin{comment}
\begin{lemma}\label{lem:partialobservability}Suppose that $B$ is dense, $\succeq'$ is complete, and each of $\succeq$ and $\succeq^*$ are continuous and locally strict complete relations.  Then if \[\succeq'|_{B\times B}\subseteq \succeq^*|_{B\times B} \cap \succeq|_{B\times B},\] it follows that $\succeq^*=\succeq$.\end{lemma}

\begin{proof}Suppose, by means of contradiction and without loss of generality, that there are $x,y\in X$ for which $x \succeq^* y$ and $y\succ x$.  By continuity of $\succeq$ and local strictness of $\succeq^*$, we can without loss of generality assume that $x \succ^* y$ and $y \succ x$.  By continuity of each of $\succeq$ and $\succeq^*$, there exists $a,b\in B$ such that $a \succ^* b$ and $b \succ a$.  But by completeness of $\succeq'$, either $a \succeq' b$, contradicting $\succeq'|_{B\times B}\subseteq \succeq|_{B\times B}$, or $b \succeq' a$, contradicting $\succeq'|_{B\times B} \subseteq \succeq^*|_{B\times B}$.\end{proof}
\end{comment}

\begin{lemma}\label{lem:supsinfsRn}
        Let $X\subseteq\Re^n$. If $\{x'_n\}$ is an increasing sequence in $X$, and $\{x''_n\}$ is a decreasing sequence, such that $\sup\{x'_n:n\geq 1\} = x^* = \inf\{x''_n:n\geq 1\}$. Then \[\lim_{n\rightarrow \infty} x'_n = x^* = \lim_{n\rightarrow \infty} x''_n.\]
\end{lemma}
\begin{proof}
        This is obviously true for $n=1$. For $n>1$, convergence and sups and infs are obtained component-by-component, so the result follows.
\end{proof}

\begin{lemma}\label{lem:squeeze}
	Let $X = \Delta([a,b])$. Let $\{x_n\}$ be a convergent sequence in $X$, with $x_n\rightarrow x^*$. Then there is an increasing sequence $\{x'_n\}$ and an a decreasing sequence $\{x''_n\}$ such that $x'_n\leq x_n\leq x''_n$, and $\lim_{n\rightarrow \infty} x'_n = x^* = \lim_{n\rightarrow \infty} x''_n$.
\end{lemma}
\begin{proof}
	The set $X$ ordered by first order stochastic dominance is a complete lattice (see, for example, Lemma 3.1 in \cite{kertz2000}). Suppose that $x_n\rightarrow x^*$. Define $x'_n$ and $x''_n$ by
	$x'_n = \inf\{x_m : n\leq m \}$ and $x''_n = \sup\{x_m : n\leq m \}$. Clearly, $\{x'_n\}$ is an increasing sequence, $\{x''_n\}$ is decreasing, and $x'_n\leq x_n\leq x''_n$.

	Let $F_x$ denote the cdf associated with $x$. Note that  $F_{x''_n} (r)  = \inf \{F_{x_m}(r) : n\leq m \}$ while $F_{x'_n} (r)$ is the right-continuous modification of $\sup \{F_{x_m}(r) : n\leq m \}$.
	For any point of continuity $r$ of $F$, $F_{x_m}(r) \rightarrow F_{x^*}(r)$, so \[
	F_x(r) = \sup \{ \inf \{F_{x_m}(r) : n\leq m \} : n\geq 1\} \] by Lemma~\ref{lem:supsinfsRn}.

	Moreover, $F_{x^*}(r) = \inf \{ \sup \{F_{x_m}(r) : n\leq m \} : n\geq 1\}$. Let $\ep>0$. Then
	\[\begin{split}
	F_{x^*} (r-\ep) \leftarrow  \sup \{F_{x_m}(r-\ep) : n\leq m \} \leq
	F_{x'_n} (r)  \leq  \sup \{F_{x_m}(r+\ep) : n\leq m \} \\ \rightarrow F_{x^*}(r+\ep)
	\end{split}	\] Then $F_{x'_n} (r) \rightarrow F_{x^*}(r)$, as $r$ is a point of continuity of $F_{x^*}$.
\end{proof}

The results we have obtained motivate two definitions that will prove useful.  Say that the set $X$, together with the collection of finite experiments $\Sigma_\infty$, has the \df{countable order property} if for each $x\in X$ and each neighborhood $V$ of $x$ in $X$ there is $x',x''\in (\cup_i B_i) \cap V$ with $x'\leq x\leq x''$.
We say that $X$ has the \df{squeezing property} if for any convergent sequence $\{x_n\}_n$ in $X$, if $x_n\rightarrow x^*$ then there is an increasing sequence $\{x'_n\}_n$, and an a decreasing sequence $\{x''_n\}_n$, such that $x'_n\leq x_n\leq x''_n$, and  $\lim_{n\rightarrow \infty} x'_n = x^* = \lim_{n\rightarrow \infty} x''_n$.

\begin{lemma}\label{lem:AAsqueezeandctbleorder} If $X=\Delta([a,b])^S$, then $X$ has the squeezing property, and there is $\Sigma_\infty$ such that $(X,\Sigma_\infty)$ has the countable order property.
 \end{lemma}
 
\begin{proof} The squeezing property follows from Lemma~\ref{lem:squeeze}, and the countable order property from  Theorem 15.11 of \citet{aliprantis2006infinite}: Indeed, 
 let $B$ be the set of probability distributions $p$ with finite support on $\Qe\cap [a,b]$, where for all $q\in\Qe\cap [a,b]$, $p(q)\in\Qe$.  Then we may choose a sequence of pairs $B_i$, and let $\Sigma_{\infty}$ to be $B_i$ with $B=\cup B_i$ so that the countable order property is satisfied. 
 \end{proof}

\subsection{Proof of Theorem~\ref{thm:convergenceutilities}}
Without loss of generality, we may set $[a,b]=[0,1]$. First we show that $u^k\to u$ in the compact-open topology. To this end, let $x^k\to x$. We want to show that $u^k(x^k)\to u(x)$. Suppose then that this is not the case, and by selecting a subsequence that $u^k(x^k)\to Y>u(x)$ (without loss). Note that $\da_{x^k}\sim^k p^k$,  where $p^k$ is the lottery that pays $1$ with probability  $u^k(x^k)\in [0,1]$, and $0$ with probability  $1-u^k(x^k)$. Let $p$ be the lottery that pays $1$ with probability $Y$, and $0$ with probability  $1-Y$ (given that the range of $u^k$ is $[0,1]$, we must have $Y\in [0,1]$). Now we have that $(\da_{x^k},p^k)\to (\da_x, p)$ and $\da_{x^k}\sim^k p^k$ implies $\da_x\sim p$. This is  a contradiction because $\da_x$ is indifferent in $\succeq$ to the lottery that pays $1$ with probability  $u^k(x^k)\in [0,1]$, and $0$ with probability  $1-u^k(x^k)$. The latter is strictly first-order stochastically dominated by the lottery $p$.  

To finish the proof, we show that $V^k\to V$. This is the same as proving that  $V^k(f^k)\to V(f)$ when $f^k\to f$. For each $k$, continuity and weak monotonicity imply that there is $x^k\in [0,1]$ so that \[
V^k(f^k) = V^k(\da_{x^k},\ldots, \da_{x^k}) = u^k(x^k).
\] Similarly, there is $x$ with $V(f) = V(\da_x,\ldots,\da_x)=u(x)$.

Now we argue that $x^k\to x$. Indeed $\{x^k\}$ is a sequence in $[0,1]$. If there is a subsequence that converges to, say, $x'> x$ then we may choose $x''=\frac{x+x'}{2}$ and eventually
\[ f^k \succeq^k (\da_{x''},\ldots,\da_{x''}) \succ (\da_{x},\ldots,\da_{x})\sim  f, \] using weak monotonicity. This is impossible because $(f^k,(\da_{x^k},\ldots,\da_{x^k})\to (f,(\da_{x'},\ldots,\da_{x'}))$ and $f^k\succeq^k ((\da_{x^k},\ldots,\da_{x^k})$ imply that  $f\succeq ((\da_{x'},\ldots,\da_{x'}) \succeq (\da_{x''},\ldots,\da_{x''})$.

Finally, using what we know about the convergence of $u^k$ to $u$, $V^k(f^k) = u^k(x^k)\to u(x)=V(f)$.

We now turn to the second statement in the theorem. Observe that $H^k$ is a continuous function from $[0,1]^S$ onto $[0,1]$.  Let $z^k \in [0,1]^S$ be an arbitrary convergent sequence, and say that $z^k \rightarrow z^*$.  We claim that $H^k(z^k)\rightarrow H(z^*)$.  Without loss we may assume that $H^k(z^k) \rightarrow Y$, by taking a subsequence if necessary.  For each $k$ and $s$, choose $y^k(s)\in [0,1]$ for which $u^k(y^k(s))=z^k(s)$. Again, without loss, we may assume that $y^k\rightarrow y^*$ by taking a subsequence if necessary, and using the finiteness of $S$.  Observe also that $u(y^*(s))=z^*(s)$ as we have shown that $u^k\to u$ in the compact-open topology.  

Now, we may also choose $\hat z^k \in [0,1]$ so that \[u^k(\hat z^k)=H^k(z^k)=H^k((u^k(y^k(s)))_{s\in S}),\] and further may again without loss (by taking a subsequence) assume that $\hat z^k$ converges to $\hat z^*$. Thus $u(\hat z^*)=\lim u^k(\hat z^k)= \lim H^k(z^k)=Y$, again using what we have shown regarding $u^k\to u$.  Then $(\delta_{\hat z^k},\ldots,\delta_{\hat z^k})\sim^k (y^k(s))_{s\in S}$ so that, by taking limits, $(\delta_{\hat z^*},\ldots,\delta_{\hat z^*})\sim^* (y^*(s))_s$.  This implies that $Y=u(\hat z^*)=H(u(y^*(s))=H(z^*)$.  

\subsection{Proof of Proposition~\ref{prop:certaintyequiv}}

Take $(\succeq^k,p^k)$ as in the statement of the Proposition, and observe that for every $p\in\Delta([a,b])$, $\mbox{ce}(\succeq^k,p^k)\in[a,b]$.  Suppose by means of contradiction that $\mbox{ce}(\succeq^k,p^k)\rightarrow \mbox{ce}(\succeq,p)$ is false, then there is some $\epsilon > 0$ and a subsequence for which $|\mbox{ce}(\succeq^{k_m},p^{k_m})-\mbox{ce}(\succeq,p)|>\epsilon$, by taking a further subsequence, we assume without loss that $\mbox{ce}(\succeq^{k_m},p^{k_m})\rightarrow \alpha\neq\mbox{ce}(\succeq,p)$.  Now, $p^{k_m}\sim^{k_m} \delta_{\mbox{ce}(\succeq^{k_m},p^{k_m})}$, and $p^{k_m}\rightarrow p$ and $\delta_{\mbox{ce}(\succeq^{k_m},p^{k_m})}\rightarrow \delta_{\alpha}$.  So by definition of closed convergence, it follows that $p\sim \delta_{\alpha}$; but this violates certainty monotonicity as $\alpha \neq \mbox{ce}(\succeq,p)$.

\section{Proof of Theorem~\ref{thm:noisychoice}}\label{sec:pfnoisythm}
First some notation. Let $\mu_n(\succeq) = \frac{1}{n}\sum_{i=1}^n \one_{x_i \succeq y_i}$, and
$\succeq_n\in \P$ be represented by $u_n\in \U$. By definition of $u_n$, we have
that $\mu_n(\succeq_n)\geq \mu_n(\succeq)$ for all $\succeq\in \P$. And we use
$\Vol(A)$ to denote the \df{volume} of a set $A$ in $\Re^d$, when this is well
defined (see \cite{schneider2014convex}). 

Consider the measure $\mu$ on $X\times X$ defined as \[
\mu(A,\succeq) = \int_A q(\succeq; x,y) \diff \lambda (x,y).
\]

In particular
\[
\mu(\succeq',\succeq) = \int_{X\times X} \one_{\succeq'}(x,y) q(\succeq; x,y) \diff \lambda (x,y).
\] is the probability that a choice with error made at a randomly-drawn choice
problem by an agent with preference $\succeq$ will coincide with $\succeq'$.

The key identification result shown in \cite{chambers2021recovering} is that, if $\succeq'\neq \succeq$, then
\[\mu(\succeq',\succeq) < \mu(\succeq,\succeq).\]

\begin{lemma}\label{lem:Lipsch}  Consider a Lipschitz noise choice environment
  $(X,\P,\la,q)$.  There is a constant $C$ with the following property. If $\succeq$ and $\succeq'$ are two preferences in $\P$ with
  representations $u$ and $u'$ (respectively) in $\U$. Then \[
  C \rho(u,u')^d \leq   \mu(\succeq,\succeq) - \mu(\succeq',\succeq) 
  \]
  \end{lemma}

\begin{proof} The ball in $\Re^d$ with center $x$ and radius $\ep$ is denoted by
  $B_\ep(x)$. First we show that the map \[\ep\mapsto \frac{\Vol(B_{\epsilon}(x)\cap X)}{\Vol(B_{\epsilon}(x))},\] defined for $x\in X$, is  nonincreasing as a function of $\epsilon>0$.

  Indeed, let $\epsilon_1 < \epsilon_2$, and  let $y\in B_{\epsilon_2}(x)\cap X$.  Then $y\in
X$ and $\|y-x\|\leq \epsilon_2$.  By convexity of $X$, $y_1\equiv
x+\frac{\epsilon_1}{\epsilon_2}(y-x)=(1-\frac{\epsilon_1}{\epsilon_2})x +
\frac{\epsilon_1}{\epsilon_2}y\in X$, and $y_1\in B_{\epsilon_1}(x)$.   Observe
further by properties of Lebesgue measure in $\Re^d$ that
$\Vol(\{x+\frac{\epsilon_1}{\epsilon_2}(y-x):y\in B_{\epsilon_2}(x)\cap
X\})=\left(\frac{\epsilon_1}{\epsilon_2}\right)^d \Vol
(B_{\epsilon_2}(x)\cap X)$.  Therefore, $\Vol(B_{\epsilon_1}(x)\cap X)\geq
\left(\frac{\epsilon_1}{\epsilon_2}\right)^d \Vol (B_{\epsilon_2}(x)\cap
X)$.  Since
$\Vol(B_{\epsilon_1}(x))=\left(\frac{\epsilon_1}{\epsilon_2}\right)^d
\Vol(B_{\epsilon_2}(x))$, it follows that\[
\frac{\Vol(B_{\epsilon_1}(x)\cap X)}{\Vol(B_{\epsilon_1}(x))}\geq
\frac{\Vol
  (B_{\epsilon_2}(x)\cap X)}{\Vol(B_{\epsilon_2}(x))},\] like we wanted to
show.

Now observe that there exists $\bar\ep >0$ large enough that $X\subseteq B_{\ep}(x)$
for all $\ep\geq \bar\ep$ and $x\in X$. Hence, for any $x\in X$ and $\ep\in(0,\bar\ep]$
\[
\frac{\Vol(B_{\epsilon}(x)\cap X)}{\Vol(B_{\epsilon}(x))}\geq
\frac{\Vol(X)}{\Vol(B_{\bar \epsilon}(x))}\equiv c' > 0, \] as $X$ has
nonempty interior and the volume of a ball in $\Re^d$ is independent of its center.

Now we proceed with the proof of the statement in the lemma. Let $\Delta =
\rho(u,u')$ and fix $x\in X$ with (wlog) 
$u(x)-u'(x)=\Delta>0$. Set
\[
\ep = \frac{\Delta}{4\kappa}.
\] We may assume that  $\ep\leq 2\bar\ep$ as defined above, as otherwise we can use a
larger upper bound on the Lipschitz constants for the functions in $\U$.

Consider the interval
\[
I=[(u'(x)+\kappa\ep)\one, (u(x)-\kappa\ep)\one], 
\] with volume
\[ 
(u(x)-\kappa\ep - (u'(x)+\kappa\ep))^d  =  (\Delta/2)^d.
\]

Consider $B_{\ep/2}(x)$. If $y\in B_{\ep/2}(x)$ then $\abs{\tilde u(y) - \tilde u(x)}< \kappa \ep$ for any $\tilde u\in \U$. 

Now, if $z\in I$ and $y\in B_\ep(x)$ then 
\[
u(y) > u(x) - \kappa \ep = u((x - \kappa \ep)\one)  \geq  u(z)
\] by monotonicity. Similarly,
\[
u'(z)\geq  u'((x+\kappa \ep)\one)= u'(x)+\kappa \ep > u'(y)
\]
Thus $(y,z)\in \succ\setminus \succeq'$ for any $(y,z)\in B_\ep(x)\times I$, and \begin{align*}
  \mu(\succeq,\succeq) - \mu(\succeq',\succeq) 
& = \int 1_{\succ\setminus \succ'}(y,z)[q(\succeq;(y,z))- q(\succeq;(z,y))]\diff
  \la(y,z) \\
& \geq \int_{B_{\ep/2}(x)\times I} 1_{\succ\setminus \succ'}(y,z)[q(\succeq;(y,z))- q(\succeq;(z,y))]\diff
  \la(y,z) \\
  & \geq \la(B_{\ep(x)/2}\times I) \inf\{q(\succeq;(y,z)- q(\succeq;(z,y)):(y,z)\in B_{\ep/2}(x)\times I \}.
  \end{align*}

Where the first identity is shown in \cite{chambers2021recovering}. The second inequality follows because
$q(\succeq;(x,y))>1/2>q(\succeq;(y,x))$ on  $(x,y)\in\succ$. The third inequality is
because $(y,z)\in \succ\setminus {\succeq}'\subseteq \succ\setminus {\succ'}$ on $B_\ep(x)\times I$.

By the assumptions we have placed on $\la$, and the calculations above, we know
that \[
\la(B_{\ep(x)/2})\geq \bar c \; \Vol(B_{\bar \epsilon}(x)\cap X) \geq \bar c
c'\; \Vol(B_{\bar \epsilon}(x)) = \bar c c' \frac{(\ep/2)^d \pi^{d/2}}{\Gamma(1+d/2)}.\]
So there is a constant $C''$ (that only
depends on $X$ and $\bar c$) so that $\la(I\times B_{\ep/2}(x))$ is bounded below
by  
\[
  (\Delta/2)^d \frac{C''(\ep/2)^d \pi^{d/2}}{\Gamma(1+d/2)}
= (\Delta/2)^d \frac{C''\Delta^d \pi^{d/2}}{(8\kappa)^d \Gamma(1+d/2)} = C' \Delta^{2d}.
\] Here $C'$ is a constant that only depends on $C''$, $\kappa$ and $d$.

By the assumption that $\Theta>1/2$,  we get that \[
  \mu(\succeq,\succeq) - \mu(\succeq',\succeq) \geq C \Delta^{2d}
\] for some constant $C$ that depends on  $C'$ and $\Theta$. 
\end{proof}

\begin{lemma}\label{lem:homothetic} Consider a homothetic noise choice environment
  $(X,\P,\la,q)$.   There is a constant $C$ with the following property. If $\succeq$ and $\succeq'$ are two preferences in $\P$ with
  representations $u$ and $u'$ (respectively) in $\U$. Then \[
  C \rho(u,u')^{2d} \leq   \mu(\succeq,\succeq) - \mu(\succeq',\succeq) 
  \]
  \end{lemma}
\begin{proof}
Let $x\in X$ be such that \[\rho(u,u')\leq u(x)-u'(x)=\Delta>0.\] Choose $\eta\in (0,1)$
so that $u(\eta x)-u'(x)=\Delta/2$. Let
\[I=(u'(x)\one,u(\eta x)\one)\] and
\[Z_\eta = [\eta x,x]\cap D^M_\al.\] 
Note that $I\subseteq X$ because by  homotheticity, $\norm{x}=M$ and hence $x\geq \al \one$. Then we must have $\al\one \leq u'(x)\one$ as $\al\one \not\leq u'(x)\one$ would mean that $u'(x)\one\ll \al\one$, contradicting monotonicity and $x\sim' u'(x)\one$. 

Observe that if $y\in I$ and $z\in Z_\eta$ then we have that
\[
u(y)< u( u(\eta x)\one ) = u(\eta x) \leq u(z),
\] as $y<u(\eta x)\one$ and $\eta x\leq z$; while
\[
u'(z)\leq u'(x) = u'( u'(x)\one) < u'(y).
\]
Hence $(z,y)\in{\succ} \setminus {\succeq'}$.

First we estimate $\Vol(Z_\eta)$. Write $Z_0$ for $[0,x]\cap D^M_{\al}$. Define the function $f(z) = x + (1-\eta)(z-x)$ and
note that when $z\in Z_0$ then $f(z)=\eta x+(1-\eta) z\in [\eta x,x]$ because $z\geq 0$. Note also that
$f(z)$ is a convex combination of $x$ and $z$, so $f(z)\in D^M_{\al}$ as the latter
is a convex set. This shows that \[Z_\eta = \{x\} + (1-\eta)(Z_0 - \{x\}),\] and
hence that $\Vol(Z_\eta) = (1-\eta)^d \Vol(Z_0)$.

Now, since $Z_0$ is star shaped we
have \[\Vol(Z_0)=\frac{1}{d} \int_{y\in  S^M_\al}\rho(y,[0,x])^d \diff y \geq (\frac{\al}{M})^d
A^M_\al,\] where $A^M_\al$ is the surface area of $S^M_\al$ and $\rho(y,[0,x])=
\max\{\ta>0:\ta y\in [0,x]$ is the radial function of the set $[0,x]$ (see \cite{schneider2014convex} page 57). The inequality results from $\rho(y,[0,x])\geq \al/M$ as $x_i\geq \al$ and $y_i\leq M$ for any $y\in S^M_\al$.

Now,
\[
1-\eta = 1- \frac{\Delta/2 + u'(x)}{u(x)}  = \frac{\Delta/2}{u(x)}\geq \frac{\Delta/2}{M},
\] as $u(x)\leq M$. Thus we have that 
\[
\Vol(Z_\eta)
\geq \Delta^d C',
\] with $C' = \Vol(Z_0)/(2M)^d>0$, a constant.

Moreover, we have $\Vol(I)= (\Delta/2)^d$ as $I\subseteq X$. Then we obtain, again
using a formula derived in  \cite{chambers2021recovering}, and that $q(\succeq;(x,y))>1/2>q(\succeq;(y,x))$ on  $(x,y)\in\succ$:
 \begin{align*}
  \mu(\succeq,\succeq) - \mu(\succeq',\succeq) 
& = \int 1_{\succ\setminus \succ'}(z,y)[q(\succeq;(z,y))- q(\succeq;(y,z))]\diff
  \la(z,y) \\
& \geq \int_{Z_\eta \times I} 1_{\succ\setminus \succ'}(z,y)[q(\succeq;(z,y))- q(\succeq;(y,z))]\diff
  \la(z,y) \\
  & \geq \la(Z_\la \times I) \inf\{q(\succeq;(z,y)- q(\succeq;(y,z)):(z,y)\in Z_\eta \times I \} \\
  & \geq (\Delta/2)^d C' \Delta^d \Theta,
 \end{align*}
 where $\Theta = \inf\{q(\succeq;(z,y)- q(\succeq;(y,z)):(z,y)\in Z_\eta \times I \}>0$.
\end{proof}

\subsection{Proof of Theorem~\ref{thm:noisychoice}}
For the rest of this proof, we denote $\mu(\succeq,\succeq^*)$ by $\mu(\succeq)$.

The rest of the proof uses routine ideas from statistical learning theory. By
standard results (see, for example, Theorem 3.1 in \cite{boucheron2005theory}), there exists an
event $E$ with probability at least $1-\da$ on which:
\[
\sup \{\abs{\mu_n(\succeq) - \mu(\succeq)} : \succeq\in P\}
\leq \E \sup \{\abs{\mu_n(\succeq) - \mu(\succeq)} : \succeq\in P\}
+ \sqrt{\frac{2\ln (1/\da)}{n}}.
\]
Moreover, again by standard arguments (see Theorem 3.2 in \cite{boucheron2005theory}), we also have
\[
\E \sup \{\abs{\mu_n(\succeq) - \mu(\succeq)} : \succeq\in P\}\leq 2 \E
\sup\{\frac{1}{n}\abs{\sum_i\sa_i \one_{\tilde x_i \succeq y_i}} : \succeq\in \P \},
\]
where \[
R_n(\P) = \E
\sup\{\frac{1}{n}\abs{\sum_i\sa_i \one_{\tilde x_i \succeq y_i}} : \succeq\in \P \}
\] is the Rademacher average of $\P$.

Now, by the Vapnik-Chervonenkis inequality (see Theorem 3.4 in
\cite{boucheron2005theory}), we have that 
\[
\E \sup \{\abs{\mu_n(\succeq) - \mu(\succeq)} : \succeq\in P\}\leq K\sqrt{\frac{V}{n}},
\] where $V$ is the VC dimension of $\P$, and $K$ is a universal constant.

So on the event $E$, we have we have that 
\[
\sup \{\abs{\mu_n(\succeq) - \mu(\succeq)} : \succeq\in P\}
\leq  K\sqrt{V/n} + \sqrt{\frac{2\ln (1/\da)}{n}}.
\]

We now combine these statements with Lemmas~\ref{lem:Lipsch} and~\ref{lem:homothetic}. In particular, we let $D=d$ or $D=2d$ depending on which of the lemmas we invoke. Let $u^*\in \U$ represent $\succeq^*$ and $u_n\in\U$ represent $\succeq_n$. Let
$\Delta = \rho(u^*,u_n)$, a magnitude that depends on the sample. Then, on the event
$E$, by Lemma~\ref{lem:Lipsch} or~\ref{lem:homothetic}, we have that
\begin{align*}
  C\Delta^D & \leq  \mu(\succeq^*) - \mu(\succeq_n) \\
  & = \mu(\succeq^*) - \mu_n(\succeq^*) + \mu_n(\succeq^*) - \mu_n(\succeq_n) + 
   \mu_n(\succeq_n)  - \mu(\succeq_n) \\
  & \leq  2K\sqrt{\frac{V}{n}} + 2\sqrt{\frac{2\ln (1/\da)}{n}},
  \end{align*} where we have used that $\mu_n(\succeq^*) - \mu_n(\succeq_n) <0$ by
definition of $\succeq_n$. This proves the second statement in the theorem.

To prove the first statement in the theorem, by Lemmas~\ref{lem:Lipsch} and~\ref{lem:homothetic} again, and using that $\mu_n(\succeq_n)\geq \mu_n(\succeq^*)$, we have that, for any $\eta>0$,
\begin{align*}
\Pr (\rho(u^*,u_n)>\eta) & \leq \Pr ( \mu(\succeq^*) - \mu(\succeq_n) >
C\eta^{D}) \\
& \leq \Pr ( \mu(\succeq^*) - \mu_n(\succeq^*) > C\eta^D/2 )
+ \Pr ( \mu_n(\succeq_n) - \mu(\succeq_n) > C\eta^D/2 ) \\
& \leq 2 \Pr ( \sup\{ \abs{\mu(\succeq') - \mu_n(\succeq')}:\succeq'\in\P\} > C\eta^D/2) \to 0
\end{align*} as $n\to\infty$ by the uniform convergence in probability result shown
in \cite{chambers2021recovering}.

\subsection{Proof of Theorem~\ref{thm:weaklymon-v1}}

By standard results (see \cite{HILDENBRAND1970161}), since $X$ is locally compact
  Polish, the topology of closed convergence is compact metric.

We will show that for any subsequence of $\succeq^k$, there is a subsubsequence converging to $\succeq^*$, which will establish that $\succeq^k \rightarrow \succeq^*$.

So choose a convergent subsubsequence of the given subsequence.  To simplify notation and with a slight abuse of notation, let us also refer to this subsubsequence as $\succeq^k$.  Call its limit $\succeq$; $\succeq$ is complete as the set of complete relations is closed in the closed convergence topology.  It is therefore sufficient to establish that $\succ^* \subseteq \succ$ and $\succeq^*\subseteq \succeq$.

First we show that $x\succ^* y$ implies that $x\succ y$. So let $x\succ^* y$.  Let $U$ and $V$ be neighborhoods of $x$ and $y$, respectively, such that $x'\succ^* y'$
for all $x'\in U$ and $y'\in V$. Such neighborhoods exist by the continuity of $\succeq^*$. We prove first that if $(x',y')\in U\times V$, then there exists $N$
such that $x'\succ_n y'$ for all $n\geq N$. Recall that $B=\cup\{B':B'\in\Sigma_{\infty} \}$. By hypothesis, there exist $x''\in U\cap B$ and $y''\in V\cap B$ such that $x''\leq x'$ and $y'\leq y''$. Each $\succeq_n$ is a strong rationalization of the finite experiment of order $n$, so if $\{\tilde x,\tilde y\}\in\Sigma_n$ then $\tilde x\succ_n \tilde y$ implies that $\tilde x\succ_m \tilde y$ for all $m\geq n$. Since $x'',y''\in B$, there is  $N$ is such that $\{x'',y''\}\in\Sigma_N$. Thus $x''\succ^* y''$ implies that  $x''\succ_n y''$ for all $n\geq N$. So, for $n\geq N$,  $x'\succ_n y'$, as $\succeq_n$ is weakly monotone.

Now we establish that $x\succ y$. Let $\{(x_n,y_n)\}$ be an arbitrary sequence with $(x_n,y_n)\rightarrow (x,y)$. By hypothesis, there is an increasing sequence $\{x'_n\}$, and a decreasing sequence  $\{y'_n\}$, such that $x'_n\leq x_n$ and $y_n\leq y'_n$ while $(x,y)= \lim_{n\rightarrow\infty }(x'_n,y'_n)$.

Let $N$ be large enough that $x'_N\in U$ and $y'_N\in V$. Let $N'\geq N$ be such that $x'_N\succ_{n} y'_N$ for all $n\geq N'$ (we established the existence of such $N'$ above). Then, for any $n\geq N'$ we have that \[x_n\geq x'_n \geq x'_N \succ_n y'_N \geq y'_n\geq y_n.\] By the weak monotonicity of $\succeq_n$, then, $x_n\succ_n y_n$. The sequence $\{(x_n,y_n)\}$ was arbitrary, so $(y,x)\notin \succeq= \lim_{n\rightarrow\infty}\succeq_n$. Thus $\neg (y\succeq x)$. Completeness of $\succeq$ implies that  $x\succ y$.

In second place we show that if $x\succeq^* y$ then $x\succeq y$, thus completing the proof.  So let $x\succeq^* y$.  We recursively construct  sequences $x^{n_k},y^{n_k}$ such that $x^{n_k}\succeq^{n_k}y^{n_k}$ and $x^{n_k}\rightarrow x$, $y^{n_k}\rightarrow y$.

So, for any $k\geq 1$, choose $x'\in N_x(1/k)\cap B$ with $x'\geq x$, and $y'\in N_y(1/k)\cap B$ with $y'\leq y$; so that $x'\succeq^* x\succeq^* y\succeq^* y'$, as $\succeq^*$ is weakly monotone. Recall that
$\succeq_n$ strongly rationalizes $c_{\succeq^*}$ for $\Sigma_n$. So $x'\succeq^* y'$ and $x',y'\in B$ imply that  $x'\succeq_n y'$ for all $n$ large enough. Let  $n_k > n_{k-1}$ (where we can take $n_0=0$) such that  $x'\succeq_{n_k} y'$; and let $x^{n_k}=x'$ and $y^{n_k}=y'$.

Then we have $(x^{n_k},y^{n_k})\rightarrow (x,y)$ and $x_{n_k} \succeq_{n_k} y_{n_k} $. Thus $x\succeq y$.

\subsection{Proof of Theorem~\ref{thm:connected-v1}}

First, it is straightforward to show that $x \succ y$ implies $x \succeq' y$.  Because otherwise there are $x,y$ for which $x\succ y$ and $y\succ' x$.  Take an open neighborhood $U$ about $(x,y)$ and a pair $(z,w)\in U\cap (B\times B)$ for which $z\succ w$ and $w\succ' z$, a contradiction.  Symmetrically, we also have $x \succ' y$ implies $x\succeq y$.

Now, without loss, suppose that there is a pair $x,y$ for which $x\succ y$ and $x \sim' y$.  By connectedness and continuity, $V=\{z: x \succ z \succ y\}$ is nonempty. Indeed if we assume, towards a contradiction that $V=\os$, then  $\{z: x\succ z\}$ and $\{z:z \succ y\}$ are nonempty open sets.  Further, for any $z\in X$, either $x\succ z$ or $z \succ y$ (because if $\neg (x\succ z)$ then by completeness $z\succeq x$, which implies that $z\succ y$).  Conclude that $\{z:x \succ z\}\cup \{z :z \succ y\}=X$ and each of the sets are nonempty and open (by continuity of the preference $\succeq$); these sets are disjoint, violating connectedness of $X$. So we conclude that $V$ is nonempty. By continuity of the preference $\succeq$, $V$ os open.  

We claim that there is a pair $(w,z)\in (V\times V)\cap (B\times B)$ for which $w \succ z$.  For otherwise, for all $(w,z)\in V\times V \cap (B\times B)$, $w \sim z$.  Conclude then by continuity that for all $(w,z)\in V\times V$, $w \sim z$.  Observe that this implies that, for any $w\in V$, the set $\{z: w\succ z \succ y\}=\varnothing$, as if $w \succ z \succ y$, we also have that $x \succeq w \succ z$, from which we conclude $x \succ z$, so that $z \in V$ and hence $z \sim w$, a contradiction.  Observe that $\{z: w\succ z \succ y\}=\varnothing$ contradicts the continuity of $\succeq$ and the connectedness of $X$ (same argument as nonemptyness of $V$; see our discussion above).

We have shown that there is $(w,z)\in (V\times V)\cap (B\times B)$ for which $w\succ z$, so that $x \succ w \succ z \succ y$.  Further, we have hypothesized that $x \sim' y$.  By the first paragraph, we know that $x \succeq' w \succeq' z \succeq' y$.  If, by means of contradiction, we have $w\succ 'z$, then $x \succ' y$, a contradiction.  So $w \sim' z$ and $w\succ z$, a contradiction to $\succeq_{B\times B}=\succeq'_{B\times B}$. \

\newpage

\bibliographystyle{econometrica}
\bibliography{identification}

\end{document}